\documentclass{article}
\usepackage{amssymb}
\usepackage[toc,page]{appendix}

\usepackage{graphicx}
\usepackage{amsmath}
\usepackage{bm}
\usepackage{hyperref}
\newtheorem{theorem}{Theorem}

\newtheorem{definition}{Definition}

\newtheorem{lemma}[theorem]{Lemma}

\newtheorem{proposition}[theorem]{Proposition}
\newtheorem{remark}[theorem]{Remark}

\newtheorem{theorem-definition}[theorem]{Theorem-definition}
\newtheorem{lemma-definition}[theorem]{Lemma-definition}
\DeclareMathAlphabet{\mathpzc}{OT1}{pzc}{m}{it}
\newenvironment{proof}[1][Proof]{\textbf{#1.} }{\ \rule{0.5em}{0.5em}}

\renewcommand {\L}{\mathcal{L}}
\newcommand {\G}{\mathcal{G}}

\newcommand {\W}{\mathcal{W}}

\newcommand {\A}{\mathcal{A}}
\newcommand {\C}{\mathcal{C}}
\newcommand {\I}{\mathcal{I}}

\newcommand {\E}{\mathcal{E}}
\newcommand {\Tot}{\mathrm{Tot}}
\newcommand {\vol}{\mathrm{vol}}
\newcommand {\SD}{\mathrm{SD}}
\newcommand {\Mat}{\mathrm{Mat}}
\newcommand {\ext}{\mathrm{ext}}
\newcommand {\GE}{\G_{\mathrm{ext}}}
\newcommand {\BE}{\B_{\mathrm{ext}}}
\newcommand {\BEX}{\B_{\mathrm{ext}\, X}}

\newcommand {\IE}{\I_{\mathrm{ext}}}

\newcommand {\T}{\mathbf{T}}

\renewcommand {\O}{\mathcal{O}}
\newcommand {\B}{\mathcal{B}}

\newcommand{\dbar}{\bar \partial}

\renewcommand{\P}{\mathbf{P}}
\newcommand{\Sym}{\mathrm{Sym}}
\newcommand{\HBF}{\mathrm{HBF}}
\newcommand{\tr}{\mathrm{tr}}

\newcommand{\id}{\mathrm{id}}

\newcommand{\End}{\mathrm{End}}
\newcommand{\SO}{\mathrm{SO}}
\newcommand{\Spin}{\mathrm{Spin}}
\newcommand{\SU}{\mathrm{SU}}
\newcommand{\red}{\mathrm{red}}
\newcommand {\dvol}{\mathrm{dvol}}

\newcommand{\Hom}{\mathrm{Hom}}
\newcommand{\Ker}{\mathrm{Ker}}

\newcommand{\Imm}{\mathrm{Im}}

\newcommand{\F}{\mathcal{F}}

\newcommand{\Cl}{\mathrm{Cl}}

\renewcommand{\sl}{\mathfrak{s}\mathfrak{l}}
\newcommand{\scal}{\mathfrak{s}\mathfrak{c}}

\renewcommand{\H}{\mathrm{H}}
\newcommand{\R}{\mathcal{R}}

\newcommand{\U}{\mathcal{U}}
\newcommand{\ddef}{\mathrm{def}}

\newcommand{\ity}{_{\infty}}
\renewcommand{\u}{\mathfrak{u}}
\newcommand{\g}{\mathfrak{g}}

\newcommand{\mat}[4]{ \left( \begin{smallmatrix}   #1 & #2 \\  #3 & #4 \end{smallmatrix}\right) }
\begin{document}

\title{A Note on Self-Dual Yang-Mills Theory}
\author{M. Movshev\\Stony Brook University\\Stony Brook, NY, 11794-3651,\\ USA } 
\date{\today}
\maketitle
\tableofcontents

\begin{abstract}
We translate the classical  Atiyah-Ward correspondence into the  L$\ity$ language. We extend the correspondence to a quasi-isomorphism  between the  algebra of the self-dual four-manifold $X$ and the  algebra of the holomorphic $BF$-theory in  the twistor space $\T(X)$. 
\end{abstract}

\section{Introduction}
The problem of reformulation of general relativity and Yang-Mills theory became an important issue after the initial successes  of twistor theory \cite{Penrose0}, \cite{Ward},\cite{AW}.

The self-dual Yang-Mills theory is a degenerate version of pure Yang-Mills theory. In this note we prove that a self-dual Yang-Mills theory is perturbatively  equivalent to a certain BF theory on the twistor space. We hope that a similar technique could  be applied to pure Yang-Mills theory, that potentially could  lead to  its interesting reformulation.

Let $X$ be a four-dimensional oriented  Riemannian  manifold. Eigen-decomposition of the Hodge $*$ operator defines the splitting of the bundle of two-forms $\Omega^2$ into the sum of self-dual and anti self-dual sub bundles $\Omega^2_++\Omega^2_-$.
Fix a vector bundle $E$ over $X$ with a  unitary connection $\nabla$. Let $G\in \Omega^2_-\u$ be an ant- self-dual two-form with values in the adjoint bundle $\u=\u(E)$. The fields in the theory are the gauge equivalence classes of pairs $(\nabla,G)$. Define the Chalmers-Siegel Lagrangian density \cite{ChSi}, \cite{Siegel} by the formula 
\begin{equation}\label{E:lagrangian}
\L_X(\nabla,G)=\tr(F_-G)\dvol
\end{equation}
 where $F_-\in \Omega^2_-\u$ is the anti self-dual part of the curvature.
We will refer to the theory, defined by the Lagrangian  (\ref{E:lagrangian}) as to self-dual Yang-Mills theory.

The Lagrangian (\ref{E:lagrangian}) admits a perturbation mentioned in \cite{Witten2} and extensively studied in \cite{Mason} and \cite{WJiang} 
\begin{equation}\label{E:lagrPURE}
\tr(GF_-)+\epsilon \tr G^2
\end{equation}
This Lagrangian is  equivalent to the Lagrangian of pure Yang-Mills theory
\begin{equation}
\frac{1}{\epsilon}\left(\tr F\wedge*F-F\wedge F\right)
\end{equation}
This explains our interest in  (\ref{E:lagrangian}).

A holomorphic BF theory is defined on a $2k+ 1$ dimensional complex manifold $M$. The manifold is equipped  with a holomorphic vector bundle $\E$. Let $\End(\E)$ be the vector bundle of  local endomorphisms of  $\E$, $\O=\O_M$ be the structure sheaf and $\omega=\omega_{M}$ be the canonical line bundle. The space of fields in this theory is a space of pairs. The first field is a differential operator of the first order $D:\E\rightarrow \Omega^{0\, 1}\E$, that satisfies Leibniz  rule $D(fe)=\dbar f e+fDe$, where $f$ is a  function and $e$ is a section. The second field $H$ is a section of $\Omega^{0\, n-2}\End(\E)\otimes \omega$. Following the standard practice, $D$  extends to an operator $D:\Omega^{0\, i}\E\rightarrow \Omega^{0\, i+1}\E$.  The operator $D^2$ is the operator of multiplication Newlander-Nirenberg tensor $F$, a section of $\Omega^{0\, 2}\End(\E)$. We define the Lagrangian density  by the formula 
\begin{equation}\label{E:BF}
\L_M(D,H)=\tr(FH)
\end{equation}

We prove perturbative equivalence of  theories (\ref{E:lagrangian}) and (\ref{E:BF}) when $X$ is a self-dual Riemannian manifold and $M$ is the corresponding twistor space $\T(X)$. We use the mathematical language of L$\ity$ algebras, that is well adapted for this purpose. 

The note is organised as follows. In Section \ref{S:Background} we fix our notations and remind the reader the  basic definitions of four-dimensional differential geometry.  In Section \ref{S:u} we formulate the  ingredients of our homological approach to twistors. As L$\ity$ equivalence could be difficult to understand for non-experts we spend  Section \ref{BVFormalism} on  explanation  of its  geometric meaning. In Section \ref{S:action} we introduce BV-actions of  BF and self-dual theories. We prove the main result in Section \ref{S:main}. We moved the  proofs of some technical statements that are used in the treatment of the main theorem to the Appendix. In particular Appendix \ref{Homotopy} contains the explicit form of the kernel of the  homotopy $\dbar^*\frac{1}{\Delta'}$.

\section{Differential-Geometric Background}\label{S:Background}

\subsection{General Facts About Self-Dual Four-Manifolds}
Let $X$ be a Riemannian four-dimensional oriented  spin manifold. Let $T=T_X$ (resp. $\Omega=\Omega_X$) be its tangent (resp. cotangent) bundle.  The vector bundle $\Omega=\Omega^1$ is a part of the De Rham complex  $(\bigoplus_{i=0}^4\Omega^i,d)$. Let $P=P_X$ be the principal $\Spin(4)$ bundle. It is  the two-sheeted cover  of the bundle of orthonormal frames in $T$.   The metric $g_{ij}$ defines Levi-Civita  connection $\nabla$ in $T$. 
\begin{remark}\label{E;fsjhfej}
If we are given a representation $L$ of $\Spin(4)$ we automatically get a $\Spin(4)$ connections in the vector bundle $L_X$  associated with $P_X$.
\end{remark}

An important $\Spin(4)$ (actually $\SO(4)$) representation is in the exterior algebra $\Lambda=\bigoplus_{i=0}^4\Lambda^i$ of the euclidean, oriented four-dimensional vector space. The components $\Lambda^i$ have the euclidean structure.  The  inner product defines the Hodge $*$-operator  $*:\Lambda^i\rightarrow\Lambda^{4-i}$ by the formula: $$a\wedge*b=(a,b)\vol.$$ As usual the canonical volume element $\vol$ has the unit length and is compatible with the orientation.   The operator satisfies $*^2=\id$. The linear space $\Lambda^2$ splits into the direct sum $$\Lambda^2=\Lambda^2_++\Lambda^2_-,$$ where $\Lambda^2_{\pm}$ are the  $\pm 1$ eigen-spaces of $*$. We call them self-dual and anti-self-dual subspaces. The  finite-dimensional algebra 
\begin{equation}\label{E:smallsdalgebra}
\A=\mathbb{R}+\Lambda^1+\Lambda^2_{-}
\end{equation}
is another example of $\Spin(4)$ representation. It appeared for the first time in \cite{CD}.

If we associate $\Lambda$  with $P_X$  we reproduce the vector bundle of differential forms $\Omega$. The same construction applied to $\Lambda^{2}_{\pm}$ leads to the bundles of self an anti self-dual forms $\Omega^2_{\pm}$. If yet we repeat it for $\A$ we get the bundle  
\begin{equation}
\A_X=\Omega^0+\Omega^1+\Omega^2_-.
\end{equation} 

Following \cite{AHS} the curvature of the metric defines self-adjoint transformation  $$\R:\Omega^2\rightarrow \Omega^2 $$ given by $$\R(e_i\wedge e_j)=\frac{1}{2}\sum R_{i,j,k,l}e_k\wedge e_l,$$ where $<e_i>$ is a local orthonormal basis of $1$-forms. The block matrix decomposition 
\begin{equation}\label{E:curv}
\R=\left( \begin{array}{cc}
A & B^*  \\
B & C  \end{array} \right)
\end{equation} relative to decomposition $\Omega^2=\Omega^2_++\Omega^2_-$ enables us to define the following components: traceless Ricci tensor $B$ and the components of the Weyl tensor $\W=\W_++\W_-$  with $\W_+=A-\frac{1}{3}\tr A$, $\W_-=C-\frac{1}{3}\tr C$.

A connection in  a vector bundle $E$ over $X$ is defined by the operator of  covariant differentiation $\nabla:\Omega^0(E)\rightarrow \Omega^1(E)$. The covariant differentiation  extends to the operator $D:\Omega^i(E)\rightarrow \Omega^{i+1}(E)$ by the formula $D(\nu\otimes s)=d(\nu)\otimes s+(-1)^{\tilde \nu}\omega\otimes \nabla s, \nu\in \Omega^i $. The operator of multiplication on the curvature $F=F_{\nabla}$ is defined as the operator $D^2\in \Omega^2(\End(E))$.

On the four-manifold $X$  connection  in the vector bundle  $E$ is said to be self-dual if its curvature $F$ is in $\Omega^2_{+}(\End(E))$.  This definition has an obvious extension to principal bundles with the structure group $G$. Let $\g$ be the Lie algebra of $G$.  A connection is self-dual if the curvature $F$ belongs to self-dual two-forms $\Omega^2_{+}(\g)$ with values in the adjoint bundle.

By definition the elliptic complex of vector bundles  
\begin{equation}\label{E:fdsfdfe}
\A_X(E)\mbox{ is equal to }\Omega^0(E)\rightarrow\Omega^1(E)\rightarrow\Omega^2_-(E).
\end{equation}
 The differential $d_{E}$ is induced from $D$. The condition $$d^2_{E}=0$$ is equivalent to self-duality of the connection. If $E$ is the adjoint  $\g_X$  of some principle  $G$-bundle, then the complex of sections of $\A_X(\g)$ is a differential graded Lie algebra.  The complex $\A_X(\g)$ is responsible for infinitesimal self-dual deformations of $\nabla$ (see \cite{AHS}). Suppose $E$ is associated with a principal $G$-bundle. Then $\A_X(E)$ is a bundle of  $\A_X(\g)$-modules.

There is the companion elliptic complex of  $\A_X(\g)$-modules 
 $$\A^*_X(E)=\Omega^2_-(E)\rightarrow \Omega^3(E)\rightarrow  \Omega^4(E).$$ The differential is also induced by $D$.

\subsection{Twistor Spaces}\label{S:TS}
In this section we collected some definitions and facts related to twistor spaces. They were borrowed from \cite{AHS}.

The $\Spin(4)$ group has two complex spinor representations $W_-,W_+$ of positive and negative chiralities. These are  two-dimensional spaces. They are equipped with $\Spin(4)$-invariant complex-linear symplectic forms and positive definite Hermitian forms.

 We obtain spinor bundles  $S_-$ and $S_+$ by the  association procedure from $W_-,W_+$ and $P_X$. The bundles automatically inherit the covariantly constant bilinear forms.  The  total spin bundle $S_-+S_+$ is a module over  the complexified     Clifford algebra bundle $\Cl$ of $\Omega^1$. If we ignore the algebra structure the bundle $\Cl$ is isomorphic to $\Omega_{\mathbb{C}}=\bigoplus_{i=0}^4 \Omega^i_{\mathbb{C}} $. The complexified  one-forms act on spinors  in the  following way 
$$\Omega^1_{\mathbb{C}}\cong \Hom(S_-,S_+)\cong \Hom(S_+,S_-).$$ 
The subbundle $\Omega^2_{-\mathbb{C}}$ consists of traceless endomorphisms of $S_-$ and the real bundle $\Omega^2_{-}$ the traceless skew-hermitian endomorphisms of $S_-$. 
The similar statements hold for $S_+$.

   The twistor space $\T=\T(X)$ is the projectivization $\P^1(S_-)$ of the spinor bundle $S_-$. It is   a real six-dimensional manifold   $\P^1$-fibered over $X$
\begin{equation}\label{E:projection}
p:\T\rightarrow X.
\end{equation}  
The bundle $\T$ has an interpretation  of a bundle of the  complex structures in the tangent bundle $T_X$ compatible with the metric.

The space $\T$ has an almost complex structure \cite{AHS}. 
  Using Riemannian connection we can split the tangent bundle $\P^1(S_-)$ into vertical and horizontal parts. On the vertical part we have the complex structure of the fibers. On the horizontal part at a point $\phi\in \P^1(S_-)_x$ over $x\in X$ we put the complex structure on $\Omega^1_x$ as follows. Multiplication on $\alpha \in \Omega^1_x$ defines the real isomorphism 
\begin{equation}\label{E:mapw}
a:\alpha\rightarrow \alpha \phi
\end{equation}
 between $\Omega^1_x$ and a complex linear space $S_{+x}$. The fibers of the projection (\ref{E:projection}) are holomorphic.

\section{The Algebra $\U$}\label{S:u}
We have already seen, that algebras with $\Spin(4)$-action (more precisely $\SO(4)$-action), e.g., $\Lambda$ and $\A$, lead to interesting differential-geometric constructions. In this section we provide  another example of an   algebra that plays an important role in the analytic geometry of   twistor spaces.

We need to refine the grading of $\A$ to a bi-grading. 
The individual summands of  $\A$ (see (\ref{E:smallsdalgebra})) have bi-grading $(0,0),(3,2),(6,4)$. It is related to the standard differential-geometric grading by the formula
\begin{equation}\label{E:grading}
\A(j)=\bigoplus_{j=k-l} \A_{k,\bm{l}}.
\end{equation}

Let $L^{\bullet\, \bullet}=\bigoplus_{i\bm{j}} L^{i\, \bm{j}}$ be a bi-graded linear space (or a sheaf). Sometimes it will be convenient to drop one of the indices. To distinguish the gradings we will be using the bold italics for the second index.
Also $$(L[s](t))^{i\, \bm{j}}\overset{\ddef}=L^{i+s\, \bm{j}+\bm{t}}.$$

The algebra $\A$ is a part of a larger  graded commutative algebra $\U$.
\begin{definition}\label{D:algz}
 The algebra $\U$ is an extension of $\A$ by an ideal with a trivial multiplication $\A^*[-11](-8)$.
We will often use the reduced grading (\ref{E:grading}) in $\U$.
 
\end{definition}

Let $\U_{\red}\subset \U$ be the direct sum $$\A(0)+(\A^*[-11](-8))(0)=\mathbb{R}+\Lambda^{2\, *}_{-}.$$ 
\subsection{The Algebra-Geometric Interpretation of $\U$ }
Let $W_+(1)=W_{+\, \P}(1)$ be the two-dimensional vector bundle over $\P^1=\P^1(W_{-})$. It is the twist of the trivial two-dimensional bundle with the fiber $W_+$ by  by the ample generator of the Picard group $\O(1)=\O_{\P}(1)$. As usual $\O=\O_{\P^1}$ stands for the structure sheaf and $\O(i)$ for the tensor powers of $\O(1)$.

We shall be interested in the  exterior algebra $\Lambda(W_+(1))$. This sheaf  and a quasi-isomorphic complex of vector bundles $D^{\bullet}$, which we shall construct shortly, are the key ingredients in our homological version of Atiyah-Ward construction.
\begin{lemma}
Trivial vector bundles over $\P^1$ with fibers $W_+\otimes W_-$ and $\Sym^2W_-$ fit into  short exact sequences
\begin{equation}\label{E:short}
\begin{split}
&0\rightarrow \left(W_+(-1) \right)^{-1} \overset{\iota_1}\rightarrow \left(W_+\otimes W_-\right)^{0}\overset{\epsilon_1}\rightarrow \left(W_+(1)\right)^{1}\rightarrow 0=\C^{\bullet\, 1}\\
&0\rightarrow \left(W_-(-1) \right)^{-1} \overset{\iota_2}\rightarrow \left(\Sym^2W_-\right)^{0}\overset{\epsilon_2}\rightarrow \left(\O(2)\right)^{1}\rightarrow 0=\C^{\bullet\, 2}
\end{split}
\end{equation}
The superscript stands for the vector bundle's index in the complex. 
\end{lemma}
\begin{proof}
Immediately follows from the short exact sequence $$\O(-1)\rightarrow W_-\rightarrow \O(1).$$
\end{proof}

The truncated complexes $D^{\bullet\, 1}=\tau^{\leq 0}\C^{\bullet\, 1}$, $D^{\bullet\, 2}=\tau^{\leq 0}\C^{\bullet\, 2}$ together with  $D^{\bullet\, 0}=\O$ form a  graded algebra $$D=\O+D^{\bullet\, 1}+D^{\bullet\, 2}$$ with $\O$-linear differential $d_{\iota}$.  Besides the obvious multiplication on $\bigoplus_{i=0}^2 \C^{0\, i}$, the only nontrivial multiplication in $D$ is $$\C^{-1\, 1}\otimes \C^{0\, 1}=W_+(-1)\otimes W_+\otimes W_-\rightarrow \Lambda^2W_+\otimes W_-(-1)\cong W_-(-1)=\C^{-1\, 2}.$$

 Let   
\begin{equation}\label{E:fddadfk}
\omega\mbox{ denote the line bundle that is  isomorphic to }\O(-4).
\end{equation}  

To simplify notations we set
\begin{equation}\label{E:gbundle}
\G=\Lambda(W_+(1))\otimes (\O+\omega).
\end{equation}

 \begin{equation}\label{E:gbundleext}
\GE=D\otimes(\O+\omega).
\end{equation}

The complexified bundle of differential forms on a holomorphic manifold $M$ has a $(p,q)$ decomposition $$\Omega^i_{\mathbb{C}}\cong \bigoplus_{i=p+q}\Omega^{p\, q}.$$
We set $\Omega^{0\, \bullet}$ to be equal to  $\bigoplus_{q}\Omega^{0\, q}$. If $\E$ is a holomorphic vector bundle on $M$ then   $\Omega^{0\, \bullet}\E$ is equipped with the action of the $\dbar$ operator.

In our applications we will be interested in the $\dbar$ operator in   algebras  $\B^{\bullet}$ and $\BE^{\bullet}$. These algebras   are equal to the space of $C^{\infty}$ sections
\begin{equation}\label{E:fatext}
\B^{\bullet}=\Gamma(\P^1,\Omega_{\P^1}^{0\, \bullet}\, \G)
\end{equation} 
and
\begin{equation}\label{E:fatextext}
\BE^{\bullet}=\Gamma(\P^1,\Omega_{\P^1}^{0\, \bullet}\, \GE).
\end{equation}

Both of them   contains a differential  ideal with zero multiplication  
\begin{equation}\label{E:ideal}
\I=\Gamma(\P^1,\Omega_{\P^1}^{0\, \bullet}\, \Lambda(W_+(1))\otimes\omega),
\end{equation} 
\begin{equation}\label{E:idealext}
\IE\overset{\ddef}{=}\Gamma(\P^1,\Omega_{\P^1}^{0\, \bullet}\, D\otimes\omega).
\end{equation}

The  group $$\Spin(4)\cong \SU_+(2)\times \SU_-(2)$$ is a subgroup of the group of  symmetries of  $\Lambda(W_+(1))(\O+\omega)$. The group $\SU_+(2)$ acts linearly on $W_+$ in $\Lambda(W_+(1))(\O+\omega)$, whereas  $\SU_-(2)$ acts by automorphisms of  $\O(1)$ and $\omega$. Similarly, we define the action of $\Spin(4)$ on $\G$. $\GE$, $\B$ and $\BE$.

The algebra $\B$ has the multiplicative  grading, described in the following table. 
\begin{equation}\label{E:gradingdiag}\begin{array}{|c|llll|}\hline

{\rm degree}&&{\rm component}&&\\\hline
0&\Omega^{0\, 0}&&+\Omega^{0\, 0}\otimes \omega&\\\hline
1&\Omega^{0\, 0}W(1)&+\Omega^{0\, 1}&+\Omega^{0\, 0}W(1)\otimes \omega&+\Omega^{0\, 1}\omega\\\hline
2&\Omega^{0\, 0}\Lambda^2W(1)&+\Omega^{0\, 1}W(1)&+\Omega^{0\, 0}\Lambda^2W(1)\otimes\omega&+\Omega^{0\, 1}W(1)\otimes\omega\\\hline
3&&\ \ \Omega^{0\, 1}\Lambda^2W(1)&&+\Omega^{0\, 1}\Lambda^2W(1)\otimes\omega\\\hline
\end{array}\end{equation}

It is the classical fact of the representation theory that any complex irreducible finite-dimensional representation of $\Spin(4)$ is isomorphic to the one in the series $$\Sym^i(W_+)\otimes \Sym^j(W_-)\, i,j\geq 0.$$ This explains  the isomorphisms 
\begin{equation}\label{E:isomorphism}
\begin{split}
&\psi:\Lambda^1_{\mathbb{C}}\rightarrow W_+\otimes W_-,\\
&\psi_{\pm}\Lambda^2_{\pm \mathbb{C}}\rightarrow \Sym^2W_{\pm}.\\
\end{split}
\end{equation}

The isomorphism 
\begin{equation}\label{E:quasUB}
\U_{\mathbb{C}}\cong H^{\bullet}(\mathbf{P}^1,\Lambda(W(1))(\O+\omega))=H^{\bullet}(\B)
\end{equation}
 follows  from the identifications (\ref{E:isomorphism}), the skew-symmetric isomorphism $W^*_-\cong W_-$ and  Serre's computation \cite{Serre} of the cohomology of the projective space:
\begin{equation}\label{cohomology}
\begin{split}
&H^0(\P^1(W),\O(j))\cong \Sym^jW^*,\ j\geq 0;\\
&H^1(\P^1(W),\O(-j-2))\cong \Sym^jW\otimes \Lambda^2(W^*),\ j\geq 0.\\
\end{split}
\end{equation}

Using  Remark \ref{E;fsjhfej} we  define the infinite-dimensional bundle $\BEX$. As $\BE$ is a free  $\A$-module we can define the operator $d^{\mathrm{naive}}_{\BE}$ in $\BEX$. The formula (\ref{E:fdsfdfe}), however, does not produce the differential that satisfies $d^2=0$ equation. To fix this we notice that $(d^{\mathrm{naive}}_{\BE})^2$ is  determined by the $C$-component of the curvature tensor  (\ref{E:curv}).  It is a symmetric section of   $\End(\Sym^2W_-)_X$. Let $\sl_2$ be the complexification of the Lie algebra of $\SU_-(2)$.  We  identify $C$ with a section of $(\sl_2\otimes \Sym^2W_{-})_X$. If  the Weyl tensor  vanishes, the section $C$ has a simple description in terms of the scalar curvature $R_X$. Let $e_1,e_2,e_3$ be a basis of $\sl_2$,  orthonormal with respect to the Killing's form. Let $\psi$ be an $\SO(4)$ (in fact $\SO(3)$) equivariant isomorphism of $\sl_2$ with $\Sym^2W_{-}$. We define an $\SO(4)$-invariant element of $\sl \otimes \Sym^2W_{-}$ as a covariantly constant section $\scal=\sum_{i=1}^3e_i\otimes \psi(e_i)$. An elementary identification shows that under the isomorphism $\End(\Sym^2W_-)_X\cong (\sl \otimes \Sym^2W_{-})_X$ induced by the symmetric square of the skew-symmetric bilinear form on $W_-$ and after a suitable normalisation of $\psi$ the curvature of the self-dual manifold  $X$ maps to $R\, \scal$. 

The element $\scal$ defines a differential operator of the first order $$\scal:\O_{\P^1}\rightarrow \Sym^2W_-.$$ 
It is $\sum_{i=1}^3 \psi(e_i)L_{e_i}$, where $L_{e_i}$ is the Lie derivative, corresponding to the element $e_i$. The image of this operator   belongs to the image of $\iota_2$ (\ref{E:short}) because the twisted tangent bundle  $T_{\P^1}(2)\cong \O(4)$ has no $\SU(2)$-invariant sections. Let $$d_{\scal}:\O_{\P^1}\rightarrow W_-(-1)$$ be the operator such that $\iota_2\circ d_{\scal}=\scal$. This operator can be extended naturally to $\BE$. We define  the operator $d_{\scal X}$ in  $\BEX$ as a product of the covariantly constant extension of $d_{\scal}$ on the scalar curvature $R$. We  also extend $\epsilon$ (\ref{E:short}) to a covariantly constant homomorphism $$\epsilon:\BEX\rightarrow \B_{X}.$$ 

Finally we can claim that the bundle $\BEX$ is equipped with two anti-commuting differentials 
\begin{equation}\label{E:dbar}
d_I=\dbar_{\P^1}
\end{equation} 
and 
\begin{equation}\label{E:dfsdf}
d_{II}=d^{\mathrm{naive}}_{\BEX}+d_{\scal}+d_{\iota}.
\end{equation}
 The operators $\dbar_{\P^1}$, $d_{\iota}$  are covariantly constant extensions of operators in $\B_{\ext}$ that have the same name. 

This follows from the evident equations: $$\dbar_{\P^1}^2=d_{\iota}^2=d^2_{\scal}=0, $$ $$\{\dbar_{\P^1}^2,d^{\mathrm{naive}}_{\BEX}\}=\{\dbar_{\P^1},d_{\scal}\}=\{\dbar_{\P^1},d_{\iota}\}=0,$$ $$d^{\mathrm{naive}\, 2}_{\BEX}+\frac{1}{2}\{d_{\scal},d_{\iota}\}=0.$$

\begin{lemma}
The kernel of the homomorphism $\epsilon: \BEX\rightarrow \B_{X}$ is closed with respect to the action of the differential $d_{I}$ and $d_{II}$.
\end{lemma}
\begin{proof}
Follows immediately  from the definition.
\end{proof}

As a corollary we deduce the following. 
\begin{proposition}
The sheaf  $\B_X$ is equipped with two anti-commuting differential $d_{I}$ and $d_{II}$. The algebra of global sections  $\Gamma(X,\B_X)$ has a structure of a bicomplex.
\end{proposition}
We set 
\begin{equation}\label{E:dT}
d_{\T} \mbox{ to be equal to }d_I+d_{II}.
\end{equation}
It satisfies

\begin{equation}\label{E:dsquare}
d_{\T}^2=0
\end{equation} 

 As we shall see shortly, the integrability of the complex structure in $\T(X)$ is equivalent to the equation  (\ref{E:dsquare}).

\subsection{Formalism of CR Structures} 
We need to remind the reader some basic elements of the theory of Cauchy-Riemann (CR) structures.

A manifold $Y$ has a CR structure if the tangent bundle $T_Y$ contains a complex subbundle $\F\subset T_Y$.  The CR structure is integrable if the space of sections $\bar {\F}\subset \F+\bar \F=\F^{\mathbb{C}}$ is closed under the bracket. A local function $f$ is said to be $\F$ holomorphic if $\xi f=0$ for any section $\xi$ of $\bar \F$.

A vector bundle $\E$ over $Y$ is CR holomorphic or $\F$ holomorphic if the gluing cocycle $c_{ij}$ is $\F$ holomorphic.

The Dolbeault complex has its analogue in the CR setting. The De Rham complex $\Omega_Y$ contains an ideal $J$, generated the subspace $J^1\subset \Omega^1$ orthogonal to $\bar \F$. It follows from integrability condition that $J$ is a differential ideal. Let $\dbar_{\F}$ be the differential in the quotient $\Omega_{\F}^{\bullet}=\Omega/J$.  If $\E$ is an $\F$-holomorphic vector bundle one can define $\Omega_\F^{\bullet}\E$ along the same lines.

The manifold $\T$ can be interpreted as a CR manifold.
 The distribution $\F$ is formed by  vertical  holomorphic  vector fields with respect to projection (\ref{E:projection}). 

The projection $\Tot(S_-)\backslash \{0\}\rightarrow \P^1(S_-)=\T$ defines the principle $\mathbb{C}^*$ bundle over $\T$ together with a  series of associated line bundle $\O_{\T}(i)$. All of them, as well as $p^*T^{\mathbb{C}}_{X}$ and  $p^*S_{\pm}(i)$, are $\F$-holomorphic.

The map (\ref{E:mapw}) is a part of   $\F$-holomorphic short exact sequence of  vector bundles 
\begin{equation}\label{E:sequence}
0\rightarrow p^*S_{+}(-1)\rightarrow p^*T^{\mathbb{C}}_X\overset{a}{\rightarrow}  p^*S_{+}(1)\rightarrow 0.
\end{equation}

The Riemann metric on $X$ can be extended to the Hermitian metric in $p^*T^{\mathbb{C}}_X$. It splits the sequence (\ref{E:sequence}). It also can be used to identify the complex normal bundle to the fibers of $p$ with $(0,1)$-part of the complexification of $p^*\Omega^{1}_X$.

It was shown in \cite{AHS} that the bundle $\O(-4)$ over $\T$ has an intrinsic meaning. It is the bundle of holomorphic volume forms 
\begin{equation}\label{E:fsdfjsn}
\omega_{\T}=\Omega_{\T}^{3\, 0}.
\end{equation}  

These observations enables us to make the  identification of global $C^{\infty}$ sections of 
\begin{equation}
\begin{split}
&\B_X\mbox{ with the space of $C^{\infty}$ sections}\\
&\mbox{ of } \Omega^{0\, \bullet}(\O+\omega_{\T}))\overset{\ddef}{=} \bigoplus_{i\geq 0}\Omega^{0\, i}+\Omega^{3\, i}
\end{split}
\end{equation}
 The differential $d_{\T}$ in $\B_X$ (\ref{E:dT}) under  this identification   transforms to the twistor  $\dbar$-operator. The Newlander-Nirenberg theorem asserts that the condition $\dbar^2=0$ is equivalent integrability of the underlying almost complex structure. So $\T$ is an analytic manifold.

We  outlined an algebraic  proof of the following classical theorem.
\begin{theorem}\cite{AHS}
Suppose $X$ is a self-dual Riemannian four-manifold. Then the canonical almost complex structure in  the twistor space  $\T(X)$ is integrable.
\end{theorem}
%The presented proof is of course equivalent to the original arguments of \cite{AHS}. 
\section{Equivalence  in BV Formalism}\label{BVFormalism}
In this section we review briefly (see \cite{BV1}) the formalism of $Q$-manifolds. It is a geometric reformulation of the classical BV-formalism \cite{SchSemi}.

Let $M$ be a finite-dimensional supermanifold equipped with odd, possibly degenerate, closed two-form $\omega$, which should not be confused with $\omega$ (\ref{E:fddadfk}) or  $\omega_{\T}$ (\ref{E:fsdfjsn}). The last two will not appear in the present section. The classical BV structure is an odd vector field $Q$, that  preserves $\omega$ and satisfies 
\begin{equation}\label{E:qq}
\{Q,Q\}=0.
\end{equation}
In addition, we require that there is $S$ such that   
\begin{equation}
dS=Q\llcorner \omega
\end{equation}
for some even function $S$. Equation (\ref{E:qq}) implies that $\frac{\partial S}{\partial Q}=0$.

If $\omega$ is not degenerate, then the Hamiltonian vector field $Q$ is determined by the Hamiltonian $S$.

A map $\psi:M\rightarrow M'$ defines a morphism between system  $(M,\omega,Q,S)$  and $(M',\omega',Q',S')$ if
\begin{align}%\label{E:mor}
&\psi^*\omega'=\omega \label{E:mor}\\
&\psi^*S'=S\\
& \psi \mbox{ is a $Q-Q'$ equivariant map.}
\end{align}
The Lie derivative corresponding to the vector field $Q$ defines a nilpotent operator in the tangent space $T_x$ at a $Q$  fixed by $Q$ point $x$, fixed by $Q$. We denote the $\mathbb{Z}_2$-graded cohomology of this complex by $HT^{\bullet}_x(Q)$.

The map $\psi$ is a local quasi-isomorphism at a point $x$ if the induced map $\psi:HT^{\bullet}_x(Q)\rightarrow HT^{\bullet}_{\psi(x)}(Q')$ is an isomorphism.

Let $U_x$ be the formal neighborhood of a $Q$-fixed point $x$. Let us choose, in addition, a system of coordinates on $U_x$,i.e.  a generating (in topological sense) subspace $V\subset \O(U_x)$. The object $(M,\omega,Q,S)$ can be restricted on $U_x$ and give an example of a formal classical $BV$ structure. The maps (\ref{E:mor}), that are compatible with the chosen system of coordinates define morphisms of  these  $BV$ structures.  The $5$-tuples $(U_x,V, \omega,Q,S)$ form a category. It contains a subcategory of objects equipped with a non degenerate symplectic form.  Such subcategory is  equivalent  to the category of finite-dimensional  L$\ity$ algebras with a non-degenerate  inner product (see \cite{ASZK}). The L$\ity$ formalism  has an advantage that it allows to conveniently work with infinite-dimensional  algebras. Such algebras correspond to infinite-dimensional mechanical systems which  is  the  basic object of study in the field theory. 

We now formulate the main reduction theorem.
\begin{proposition}\label{P:quasimod}
Let $B$ be an L$\ity$ algebra. Suppose that   $B^{\bullet}$ is the diagonal complex of   the bi-complex $B^{\bullet\, \bullet}$. Let $\pi_{I}$ be a degree zero projector in $B$ such that $\id-\pi=\{H,d_I\}$, where $d_I$ and $d_{II}$ are the two anti-commuting  differentials. The  homotopy $H$ has the  bi-degree $(-1,0)$ and, in addition, satisfies 
\begin{equation}\label{E:inner}
H^2=0,\, (Ha,\pi b)=0,\mbox{ and }(Ha,b)=(-1)^{\tilde a}(a,Hb).
\end{equation}  

Additionally, we assume that the operator $H\circ d_{II}$ is nilpotent.

Then $A\overset{\ddef}=\Imm \pi$ has a structure of L$\ity$-algebra with an inner product, induced by inclusion. This algebra is quasi-isomorphic to $B$. The quasi-isomorphism is compatible with the inner-product.
\end{proposition}
\begin{proof}
The proof goes along the same lines as for the ordinary complex (see \cite{KS}, \cite{MArkl}). We only modify the structure of the trees. The trees are allowed to have  vertices of valence two. We associate with them the operator $d_{II}$. Additional care should be taken about the compatibility of inner products. In the geometric language it corresponds to equation (\ref{E:mor}). The compatibility follows from equation (\ref{E:inner}).
\end{proof}

This proposition has a simple  geometric  interpretation. Let $U_0$ be the formal neighborhood of zero in the linear odd-symplectic  super-space $B$, and $\tilde U_0$ is the similar neighborhood in $A$.    The projection $\pi$ defines a fibration  $\pi:U_0\rightarrow \tilde U_0$. The restriction of $S$ on the fiber $\pi^{-1}z, z\in \tilde U_0$  has a unique nondegenerate  critical point $\psi(z)$. This defines a section $\psi:\tilde U_0\rightarrow U_0$, whose Taylor coefficients  we interpret as the quasi-isomorphism. The compatibility relations  follow from the formal properties of the Legendre transform.

\section{Holomorphic BF  Theories}\label{S:action}
In this section we define the holomorphic BF theory data. This is the BV reformulation of the BF theory mentioned in the introduction.

We fix  an odd dimensional complex manifold $M$ together with a holomorphic vector bundle $\E$.
Let  $\End(\E)$ be the vector bundle local endomorphisms of  $\E$, $\O=\O_M$ be the structure sheaf and $\omega=\omega_{M}=\Omega^{\dim(M)\, 0}$ be the canonical line bundle on $M$.

The Dolbeault complex  
\begin{equation}\label{E:xcsja}
\Omega^{0\, \bullet}\End(\E)\otimes(\O+\omega)
\end{equation}
defines a sheaf of   graded algebra. The subsheaf $\Omega^{0\, \bullet}\End(\E)\otimes\omega$ is an ideal with zero multiplication.

The space of $C^{\infty}$ sections of  (\ref{E:xcsja}) is equipped with the  linear functional  $$\Gamma(M,\Omega^{0\, \bullet}\End(\E)\otimes\omega) \overset{\tr}{\rightarrow} \Gamma(M,\Omega^{0\, n}\omega)\overset{\int}{\rightarrow} \mathbb{C}$$ that for any two sections satisfies the identity $\int \tr(ab)=(-1)^{\tilde a \tilde b}\int \tr(ba)$.
If $M$ is compact or the sections   decay fast at infinity, then $\int$ is a $\dbar$-closed linear functional.
By definition the triple
\begin{equation}\label{E:triple}
\HBF(M,\E)=(\Omega^{0\, \bullet}\End(\E)\otimes(\O+\omega),\int\tr,\dbar)
\end{equation}
 comprise the holomorphic BF data.  It is a special case of    a Chern-Simons triple \cite{Movsuperq}. The Lagrangian density corresponding  to (\ref{E:triple}) is equal to 
\begin{equation}\label{E:fdsqaa}
\L_{\HBF}=\tr(\frac{1}{2}a \dbar a+\frac{1}{6}a^3), a\in \Gamma(M,\Omega^{0\, \bullet}\End(\E)\otimes(\O+\omega)).
\end{equation}

The reader should think about (\ref{E:fdsqaa}) as of a BV formulation of a suitable  theory,  which we also denote by  $\HBF$. More precisely the function $S$ is $\int\L_{\HBF}$, the linear and quadratic Taylor coefficients of the vector field $Q$ are the differential $\dbar$ and the graded commutator $[.,.]$. The inner product $\int\tr(ab)$ defines a constant odd symplectic form on the $\mathbb{Z}_2$-graded linear space $\bigoplus_i \Omega^{0\, i}\End(\E)\otimes(\O+\omega)$.  The classical master equation is a direct corollary of the axioms of differential graded algebra with a trace.

We can treat the sheaf of algebras $\U_X(\End(E))$ along the same line and define the  triple $$\SD(X,E)=(\L_{\SD}, \U_X(\End(E)), \int \tr)$$ with $\L_{\SD}$ equal to 
\begin{equation}
\L_{\SD}=\tr(\frac{1}{2}a d_E a+\frac{1}{6}a^3), a\in \Gamma(X,\U_X\otimes \End(E))
\end{equation}
The denote the corresponding formal BV system by $\SD(X,E)$.
\section{Equivalence of $\HBF(M,\E)$ with $\SD(X,E)$}\label{S:main}
Our main example of a holomorphic BF theory is a  theory on a twistor space $\T(X)$.
Let $E$ be a self-dual $k$-dimensional vector bundle on $X$. The Atiyah-Ward correspondence  defines a holomorphic structure on $\E$ equal to  the pullback $p^* E$ with respect to the map (\ref{E:projection}).

In this section we  construct  the quasi-isomorphism
\begin{equation}\label{E:qiso}
\Gamma(X,\U_X(\End(E)))\rightarrow \Gamma(X,\B_X\otimes \End(E))\overset{s}\cong \Gamma(\T(X),\Omega^{0\, \bullet}(\O+\omega_{\T})\otimes \End(\E))
\end{equation} 
It existence  follows from the  general Proposition \ref{P:quasimod}. To make the connection with Proposition \ref{P:quasimod} more explicit we recall that $\Gamma(X,\B_X\otimes \End(E))$ is a bicomplex with differentials $d_{I}$ and $d_{II}$, which are $E$-twisted versions of  (\ref{E:dbar}, \ref{E:dfsdf}). We also recall that the operator $\H$, constructed in Appendix \ref{Homotopy} commutes with the $\SU(2)$ action. We extend $\H\otimes \id_{\Mat_k}$ to the covariantly constant endomorphism $\H_X$ of $\B_X\otimes \End(E)$. The conditions of  Proposition \ref{P:quasimod} for this choice of homotopy are satisfied. 

The linear part $\psi_1$ of the quasi-isomorphism $\psi$ is equal to $\sum_{k\geq0} (\H_X\circ d_{II})^k\circ i$.

Proposition \ref{P:quasimod}  enables us to construct the desired quasi-isomorphism. The inclusion $i$ coincides with the pullback $p^*$. In particular  the map $\psi_1$ is equal to $p^*+\H_X\circ d_{\B_{X}\End(E)}\circ p^*$. 

We deduce the following.
\begin{theorem}\label{P:quasi}
There is a quasi-isomorphism $\psi$ (\ref{E:qiso}) of the differential graded algebras.
It  also  defines an L$\ity$ quasi-isomorphism of the graded Lie algebras  $\Gamma(X,\U_X(\End(E)))$  and $\Gamma(\T(X),\Omega^{0\, \bullet}(\O+\omega_{\T})\otimes \End(\E))$ with the bracket equal to the graded commutator.
\end{theorem}

The category of self-dual vector bundles is closed with respect to the tensor multiplication.  Let $Z^{\bullet}$ be a complex of self-dual vector bundles with covariantly constant differentials. The Atiyah-Ward correspondence  provides us with  a complex of holomorphic vector bundles $\mathcal{Z}^{\bullet}$ on $\T(X)$.  The following is an extension of the last theorem.

\begin{theorem}\label{P:quasi2}
There is a quasi-isomorphism of the differential graded modules
\begin{equation}\label{E:qiso3}
\Gamma(X,\U_X(\End(E)\otimes Z^{\bullet} ))\rightarrow \Gamma(\T(X),\Omega^{0\, \bullet}(\O+\omega_{\T})\otimes \End(\E)\otimes \mathcal{Z}^{\bullet})
\end{equation}
compatible with the quasiisomorphism (\ref{E:qiso}). 
\end{theorem}

\appendix

\section{Explicit Formula for the Kernel of $\H$}\label{Homotopy}

Recall that the complex linear spaces $W_-,\, W_+$ are equipped with positive  Hermitian inner-products (see Section \ref{S:TS}). These inner-products allow to define  the Fubini-Study  metric on $\P^1(W_-)$ and the Hermitian structure in the vector bundle 

In their presence  we can define 
$\dbar_{\P^1}^*$ and  the Laplace operator 
\begin{equation}\label{E:laplace}
\Delta=\{\dbar_{\P^1},\dbar_{\P^1}^*\},
\end{equation} 
that acts in the space of smooth sections $\bigoplus_{i=0,1}\Gamma(\P,\Omega^{0\, i}\G).$
Let $$\pi:\bigoplus_{i=0,1}\Gamma(\P,\Omega^{0\, i}\G)\rightarrow \bigoplus_{i=0,1}\Gamma(\P,\Omega^{0\, i}\G) $$ be the orthogonal projection on $\Ker\, \Delta$. 

We are interested in the homotopy 
\begin{equation}\label{E:fdsddsg}
\H:\Gamma(\P,\Omega^{0\, 1}\G)\rightarrow \Gamma(\P,\Omega^{0\, 0}\G)
\end{equation} that  satisfies 
\begin{equation}\label{E:chain}
\id-\pi=\{\dbar_{\P},\H\}.
\end{equation}

The goal of this section to write an explicit formula for the kernel of the operator of homotopy $\H$ (\ref{E:fdsddsg}). 

 The line bundle $\O(-2)$ on $\P^1$ is isomorphic to the line bundle of holomorphic one-forms $\Omega^{1\, 0}$. Thus we can identify $\O(-n)$ with $\Omega^{\frac{n}{2}\, 0}$.

The homotopy $\H_{\G}$ is the direct sum of the homotopies $\H_{\O(n)}$ for suitable $n$. It  suffices to compute the homotopies  
$\H_{-\frac{n}{2}}:\Gamma(\P^1,\Omega^{-\frac{n}{2}\, 1})\rightarrow \Gamma(\P^1,\Omega^{-\frac{n}{2}\, 0})$
for all $n$.  Then  the formula for $\H_{\G}$ will follow.

Let 
\begin{equation}\label{E:gdcdj}
\Delta=\Delta_{-\frac{n}{2}} \mbox{ be the Laplace operator in }\bigoplus_{i=0,1}\Omega^{-\frac{n}{2}\, i}.
\end{equation}
The self-adjoint  Green's operator $G=G_{-\frac{n}{2}}$  satisfies $G\Delta=\id-\pi$. It can be used to construct the homotopy $\H=\H_{-\frac{n}{2}}$. Indeed if we set 
\begin{equation}\label{E:homotop1}
\H=\dbar^*G,
\end{equation} then the identity (\ref{E:chain}) would follow automatically. As the  metric on $\P$ and the Hermitian structure on $\O(n)$ have  $\SU(2)$-symmetry, the operators $G$ and $\H$ commute with the $\SU(2)$-action.
Additionally, equation (\ref{E:inner}) follows automatically from elementary properties of the Green's operator.

The kernel $g=g_{-\frac{n}{2}}$ of the Green's operator $G_{-\frac{n}{2}}$ is a $\SU(2)$-invariant generalised section of $$\Omega^{\frac{n}{2}+1\, 1}\boxtimes\Omega^{-\frac{n}{2}\, 0}+\Omega^{\frac{n}{2}+1\, 0}\boxtimes\Omega^{-\frac{n}{2}\, 1}.$$  
It is real analytic away of the diagonal and has $\log$-singularity at the diagonal.

The kernel $h=h_{-\frac{n}{2}}$ of the operator $\H_{-\frac{n}{2}}$ coincides with the generalised section 
\begin{equation}\label{E:hdef}
\id\boxtimes\dbar^* g_{-\frac{n}{2}}\mbox{ of }\Omega^{\frac{n}{2}+1\, 0}\boxtimes\Omega^{-\frac{n}{2}\, 0}.
\end{equation}
It is $\SU(2)$-invariant by construction.

 Let $s_{-\frac{n}{2}\, \alpha}$ be a basis in the subspace of harmonic elements of $\Omega^{-\frac{n}{2}\, 0}$ and $t_{\frac{n}{2}+1}^{\alpha}$ be the dual basis in the harmonic subspace of  $\Omega^{\frac{n}{2}+1\, 1}$. The section  $ \sum_{\alpha}t^{\alpha}_{\frac{n}{2}+1} \boxtimes s_{-{\frac{n}{2}}\, \alpha}+\sum_{\beta}s_{{\frac{n}{2}
+1}\, \beta}\boxtimes t^{\beta}_{-\frac{n}{2}} $
 is a kernel of $\pi$.
The section $h$ satisfies 
\begin{equation}\label{E:kerequat}
\dbar h=\sum_{\alpha}t^{ \alpha}_{\frac{n}{2}+1} \boxtimes s_{-\frac{n}{2}\, \alpha}+\sum_{\beta}s_{\frac{n}{2}+1\, \beta} \boxtimes t^{\beta}_{-\frac{n}{2}}\mbox{ off the diagonal}.
\end{equation} This is the equation (\ref{E:chain}) written in the language of kernels. 
As the line bundle  $\O(n)$ has the trivial first cohomology for $n\geq -1$, the sections $t^{\alpha}_{-\frac{n}{2}}$ are zero. Similarly, for $n\leq -1$ the elements $s_{-\frac{n}{2}\, \alpha}$ vanish.
 In the following we will refer to $A$ as  the first and to $B$ as the second multiples in the tensor product $A\otimes B$. Then  equation (\ref{E:kerequat}) implies the following.

\begin{lemma}\label{L:hol}
If $n\geq -1$  the section $h_{-\frac{n}{2}}$ is holomorphic in the second argument. Similarly, for $n\leq -1$ $h_{-\frac{n}{2}}$ is holomorphic in the first argument. 
\end{lemma}

Let us transfer the spherical metric from $\P$ to $\mathbb{C}$ using the stereographic projection. In the holomorphic coordinate the metric is equal to 
\begin{equation}\label{E:metric}
\frac{|dz|^2}{(1+|z|^2)^2}.
\end{equation}
 The group $\SU(2)=\{\mat{a}{b}{-\bar b}{\bar a}|\, |a|^2+|b|^2=1\}$ acts birationally on $\mathbb{C}$  by M\"{o}bius transformations $f(z)=\frac{az+b}{-\bar bz+\bar a}$. These transformations  preserve the metric  (\ref{E:metric}). 
\begin{lemma}\label{L:fdjdn}
We let 
\begin{equation}
h_{-\frac{n}{2}}\mbox{ to be equal to }h_{-\frac{n}{2}}(z_1,z_2)\sqrt{dz_1}^{\,n+2}\sqrt{dz_2}^{\, -n}.
\end{equation} 
If $h_{-\frac{n}{2}}$ is $\SU(2)$ invariant, then $h_{-\frac{n}{2}}(z_1,z_2)$ satisfies 
\begin{equation}\label{E:invariant}
h_{-\frac{n}{2}}(z_1,z_2)=\frac{(-\bar bz_2+\bar a)^{n}}{(-\bar bz_1+\bar a)^{2+n}}h_{-\frac{n}{2}}\left(\frac{az_1+b}{-\bar bz_1+\bar a},\frac{az_2+b}{-\bar bz_2+\bar a}\right)
\end{equation}
\end{lemma}
\begin{proof}
Follows from the explicit computation.
\end{proof}

If we set $az_1+b=0$, the condition $|a|^2+|b|^2=1$ would imply that $|a|^2=\frac{1}{|z_1|^2+1}$. It allows to show that $$h_{-\frac{n}{2}}(z_1,z_2)=\frac{e^{2\pi i \theta}(\bar z_1 z_2+1)^{n}}{(|z_1|^2+1)^{n+1}}h_{-\frac{n}{2}}\left(0,e^{2\pi i \theta}\frac{z_2-z_1}{\bar z_1 z_2+1}\right),$$ where $\theta$ is real and satisfies $e^{2\pi i \theta}=\frac{a}{\bar a}$.
\begin{lemma}\label{L:sjsiue}
Let $f(z)$ be a function that is holomorphic in a punctured neighborhood of zero. If $f(z)$ satisfies $e^{2\pi i \theta}f(e^{2\pi i \theta}z)=f(z)$, then $f(z)$ is proportional  $\frac{1}{z}$.
\end{lemma}
\begin{proof}
Follows from the uniqueness of the Laurent series expansion.
\end{proof}

The combination of Lemmas \ref{L:hol}, \ref{L:fdjdn} and \ref{L:sjsiue} results in the following.
\begin{proposition}
Let $h_{-\frac{n}{2}}$ be the kernel of the operator $\H=\H_{-\frac{n}{2}}$, defined by the formula (\ref{E:homotop1}) using the Green's function $G$ of the Laplace operator (\ref{E:gdcdj}). Additionally we assume that  the metric in $\Omega^{-\frac{n}{2}\, 0}$ has  $\SU(2)$ symmetry.

If $n\geq -1$, then the kernel  $h_{-\frac{n}{2}}$
 is equal to  $$\frac{1}{2\pi \sqrt{-1}}\left(\frac{\bar z_1 z_2+1}{|z_1|^2+1}\right)^{n+1}\frac{\sqrt{dz_1}^{\, n+2}\sqrt{dz_2}^{\, -n}}{z_2-z_1}.$$
The use of constant $\pi$ should not produce a confusion with the projector.
If $n\leq -1$,  then  $h_{-\frac{n}{2}}$ is 
$$\frac{1}{2\pi \sqrt{-1}}\left(\frac{|z_2|^2+1}{z_1 \bar z_2+1}\right)^{n+1}   \frac{\sqrt{dz_1}^{\, n+2}\sqrt{dz_2}^{\, -n}}{z_2-z_1}.$$
\end{proposition}

Finally we can setup the formula for the kernel of $\H_{\G}$. Taking in to account  the isomorphism $\O(n)\cong \Omega^{-\frac{n}{2}\, 0}$ we can claim that $2\pi \sqrt{-1} h_{\G}$ is equal to 
\begin{equation}
\begin{split}
&\quad \left(\frac{\bar z_2 z_1+1}{|z_2|^2+1}\right)^{3}\frac{\sqrt{dz_1}^{\, -2}\sqrt{dz_2}^{4}}{z_2-z_1}+\\
&+\left(\frac{\bar z_2 z_1+1}{|z_2|^2+1}\right)^{2}\frac{\sqrt{dz_1}^{\, -1}\sqrt{dz_2}^{3}}{z_2-z_1}\id_{W}+\\
&+\left(\frac{\bar z_2 z_1+1}{|z_2|^2+1}\right)\frac{\sqrt{dz_2}^{2}}{z_2-z_1}+\\
&+\left(\frac{\bar z_1 z_2+1}{|z_1|^2+1}\right)\frac{\sqrt{dz_1}^{2}}{z_2-z_1}+\\
&+\left(\frac{\bar z_1 z_2+1}{|z_1|^2+1}\right)^{2}\frac{\sqrt{dz_1}^{3}\sqrt{dz_2}^{\, -1}}{z_2-z_1}\id_{W}+\\
&+\left(\frac{\bar z_1 z_2+1}{|z_1|^2+1}\right)^{3}\frac{\sqrt{dz_1}^{\, 4}\sqrt{dz_2}^{\, -2}}{z_2-z_1}
\end{split}
\end{equation}

\bibliographystyle{plain.bst}
\bibliography{sdym.bib}

\end{document}